\documentclass[preprint]{elsarticle}
\usepackage[T1]{fontenc}
\usepackage{pgf}
\usepackage[latin1]{inputenc}
\usepackage[english]{babel}
\usepackage{amsmath,enumerate}
\usepackage{amssymb}
\usepackage{amscd}
\usepackage{tabularx}
\usepackage{vmargin}
\usepackage{url}

\begin{document}

\begin{frontmatter}
\title{Incidence coloring of graphs with high maximum average degree \footnote{Work partially supported by the ANR Grant EGOS (2012-2015) 12 JS02 002 01.}
}
\author[LIRMM]{Marthe Bonamy}
\author[LaBRI]{Herv\'e Hocquard}
\author[LIFORCE]{Samia Kerdjoudj}
\author[LaBRI]{Andr\'e Raspaud}
\address[LIRMM]{Universit\'e Montpellier 2 - LIRMM 161 rue Ada 34095 Montpellier, France}
\address[LaBRI]{LaBRI (Universit\'e de Bordeaux), 351 cours de la Lib\'eration, 33405 Talence Cedex, France}
\address[LIFORCE]{LIFORCE, Faculty of Mathematics, USTHB, BP 32 El-Alia, Bab-Ezzouar 16111, Algiers, Algeria}

\begin{abstract}
An incidence of an undirected graph G is a pair $(v,e)$ where $v$ is a vertex of $G$ and $e$ an edge of $G$ incident with $v$. Two incidences $(v,e)$ and $(w,f)$ are adjacent if one of the following holds: (i) $v = w$, (ii) $e = f$ or (iii) $vw = e$ or $f$. An incidence coloring of $G$ assigns a color to each incidence of $G$ in such a way that adjacent incidences get distinct colors. 
In 2005, Hosseini Dolama \emph{et al.}~\citep{ds05} proved that every graph with maximum average degree strictly less than $3$ can be incidence colored with $\Delta+3$ colors. Recently, Bonamy \emph{et al.}~\citep{Bonamy} proved that every graph with maximum degree at least $4$ and with maximum average degree strictly less than $\frac{7}{3}$ admits an incidence $(\Delta+1)$-coloring. In this paper we give bounds for the number of colors needed to color graphs having maximum average degrees bounded by different values between $4$ and $6$. In particular we prove that every graph with maximum degree at least $7$ and with maximum average degree less than $4$ admits an incidence $(\Delta+3)$-coloring. This result implies that every triangle-free planar graph with maximum degree at least $7$ is incidence $(\Delta+3)$-colorable. We also prove that every graph with maximum average degree less than 6 admits an incidence $(\Delta + 7)$-coloring. More generally, we prove that $\Delta+k-1$ colors are enough when the maximum average degree is less than $k$ and the maximum degree is sufficiently large.
 \end{abstract}

 \begin{keyword}
Graph coloring -  Incidence coloring - Maximum average degree - Incidence chromatic number
 \end{keyword}

\end{frontmatter}

\newenvironment{proof}{\par \noindent \textbf{Proof} \\}{\hfill$\Box$}

\newtheorem{corollary}{Corollary}
\newtheorem{definition}{Definition}
\newtheorem{question}{Question}
\newtheorem{problem}{Problem}
\newtheorem{proposition}{Proposition}
\newtheorem{theorem}{Theorem}
\newtheorem{lemma}{Lemma}
\newtheorem{conjecture}{Conjecture}
\newtheorem{sketch}{Sketch of proof}
\newtheorem{observation}{Observation}
\newtheorem{remark}{Remark}
\newtheorem{claim}{Claim}
\newtheorem{example}{Example}

\section{Introduction}

In the following we only consider simple, non-empty connected graphs. In a graph $G=(V,E)$, an \emph{incidence} is an edge $e$ coupled with one of its two extremities, denoted by $(u,uv)$ or $(v,uv)$. In other words, incidences of $G$ are in natural bijection with edges in the graph $G_s$ obtained from $G$ by subdividing every edge once.

The set of all incidences in $G$ is denoted by $I(G)$, where
\begin{center}
 $I(G)=\{(v,e)\in V(G)\times E(G): \mbox{edge $e$ is incident to $v$}\}$.
\end{center}
Two incidences $(u,e)$ and $(v,f)$ are adjacent if one of the following holds :
\begin{center}
$i)$ $u=v$, $ii)$ $e=f$ and $iii)$ the edge $uv=e$ or $uv=f$.
\end{center}


A \emph{strong incidence of} a vertex $u$ is an incidence $(u,uv)$ for some $v$. A \emph{weak incidence of} a vertex $u$ is an incidence $(v,uv)$ for some $v$. In both cases, a strong or weak incidence of $u$ can be referred to as an \emph{incidence of} $u$.

A proper \emph{incidence coloring} of $G$ is a coloring of the incidences in such a way that for every vertex $u$, a strong incidence of $u$ does not receive the same color as any other incidence of $u$. 

This corresponds to an edge coloring of $G_s$ such that two incident edges receive different colors, and no edge is incident with two edges of the same color. We say that $G$ is \emph{incidence $k$-colorable} if it can be properly incidence colored using only integers between $1$ and $k$. Note that this definition allows for a non-integer value of $k$, though $\lfloor k \rfloor$ could be equivalently considered. We denote by $\chi_i(G)$ the {\it incidence chromatic number} of $G$, which is the smallest integer $k$ such that $G$ is incidence $k$-colorable. The notion of incidence coloring was first introduced by Brualdi and Massey~\cite{bm93}. And they posed the Incidence Coloring Conjecture, which states that: 

\begin{conjecture}[\label{conjic}Brualdi and Massey~~\cite{bm93}]
For every graph $G$, $\chi_i(G)\leq\Delta(G)+2$.
\end{conjecture}

When Brualdi and Massey defined this variant of coloring, they also provided tight bounds for its corresponding chromatic number, as follows.

\begin{theorem}[Brualdi and Massey~\cite{bm93}]\label{borne}
For every graph $G$, $\Delta(G)+1 \leq\chi_i(G)\leq 2\Delta(G)$.
\end{theorem}

However, in 1997, by observing that the concept of incidence coloring is a particular case of directed star arboricity introduced by Algor and Alon~\cite{alg}, Guiduli~\cite{gui} disproved the Incidence Coloring Conjecture showing that Paley graphs have an incidence chromatic number at least $\Delta+ \Omega(\log\Delta)$. He also improved the upper bound proposed by Brualdy and Massey in Theorem~\ref{borne}.

\begin{theorem}[Guiduli~\cite{gui}]
For every graph $G$, $\chi_i(G)\leq \Delta(G)+ o(\log\Delta(G))$.
\end{theorem}

For technical purpose, we use the stronger notion of incidence coloring introduced by Hosseini Dolama, Sopena and Zhu in 2004~\cite{hdsz04}, defined as follows. An incidence $(k,\ell)$-coloring of $G$ is an incidence $k$-coloring such that for every vertex $u$, at most $\ell$ different colors can appear on weak incidences of $u$. In particular, one can note that an incidence $(k,1)$-coloring of a graph is actually tantamount to a square coloring of it, that is, a proper coloring of its vertices with the additionnal property that no vertex can have two neighbors with the same color. Indeed, we can consider the unique color used on the weak incidences of a vertex to be assigned to that vertex, and conversely. Note again that the notion of incidence $(k,\ell)$-coloring holds for non-integer values of $k$ and $\ell$.

Let ${\rm mad}(G)=\max\left\{\frac{2|E(H)|}{|V(H)|},\;H \subseteq
G\right\}$ be the maximum average degree of the graph $G$, where
$V(H)$ and $E(H)$ are the sets of vertices and edges of $H$,
respectively. This is a conventional measure of sparsness of arbitrarily graphs (not necessary planar). For more details on this invariant see \cite{Toft} where properties of the maximum degree are exhibited and where it is proved that maximum average degree may be computed by a polynomial algorithm. Results linking maximum average degree and incidence coloring date back to 2005, where Hosseini Dolama and Sopena~\cite{ds05} started looking for such relationships in the case of graphs with low maximum average degree (\emph{i.e.} not only bounded, but bounded by a small constant). However, earlier theorems have implications on graphs with bounded maximum average degree (\emph{i.e.} bounded by any constant).
\medskip

\begin{theorem}\cite{hdsz04}\label{th:degen}
Let $k \in \mathbb{N}$, and $G$ be a $k$-degenerate graph. Then $G$ is incidence $(\Delta(G)+2k-1,k)$-colorable.
\end{theorem}

As a corollary, it holds immediately that for every integer $k$, a graph $G$ with ${\rm mad}(G)<k$, being $(k-1)$-degenerate, is $(\Delta(G)+2k-3,k-1)$-colorable. By allowing a lower bound on the maximum degree $\Delta(G)$, we seek to reduce the number of colors necessary for an incidence coloring and we prove the following result.

\begin{theorem}\label{th:BH1}
Let $k \in \mathbb{N}$, and $G$ be a graph with maximum degree $\Delta(G)$ and maximum average degree ${\rm mad}(G)<k$.
\begin{enumerate}
\item If $\Delta(G) \ge \frac{k^2}{2}+\frac{3k}{2}-2$, then $G$ is incidence $(\Delta(G)+k-1,k-1)$-colorable.
\item For all $\alpha>0$, if $\Delta(G) \ge \frac{3\alpha+1}{2 \alpha} k-2$, then $G$ is incidence $(\Delta(G)+(1+\alpha) k-1,(1+\alpha) k-1)$-colorable.
\end{enumerate}
\end{theorem}

When considering small values of $k$, Theorem~\ref{th:BH1}.1 cannot compete with specific results. 

Let us now discuss some such specific cases. The following results were proved for small values of maximum average degree.\\

\begin{theorem}\label{resume:mad}
$ $
\begin{enumerate}
  \item If $G$ is a graph with $\rm{mad}(G)<3$,  then $\chi_i(G)\leq \Delta(G)+3$.~\cite{ds05}
  \item If $G$ is a graph with $\rm{mad}(G)<3$ and $\Delta(G)\geq5$, then $\chi_i(G)\leq \Delta(G)+2$.~\cite{ds05} 
  \item If $G$ is a graph with $\rm{mad}(G)<\frac{22}{9}$, then $\chi_i(G)\leq \Delta(G)+2$.~\cite{ds05}
  \item If $G$ is a graph with $\Delta(G) \ge 4$ and $\rm{mad}(G)<\frac{7}{3}$, then $\chi_i(G)=\Delta(G)+1$.~\cite{Bonamy}
\end{enumerate}
\end{theorem}

\medskip
Let us recall results concerning the incidence chromatic number of $3$-degenerate graphs.
\begin{theorem}[\label{3degene}Hosseini Dolama and Sopena~\cite{ds05}]
Every $3$-degenerate graph $G$ admits an incidence $(\Delta(G)+4,3)$-coloring.
Therefore, $\chi_i(G)\leq \Delta(G)+4$.
\end{theorem}
Since a graph $G$ with $\rm{mad}(G)<4$ is $3$-degenerate, the following corollary can easily be derived from Theorem~\ref{3degene}:

\begin{corollary}\label{mad}
If $G$ is a graph with $\rm{mad}(G)<4$, then $G$ admits an incidence $(\Delta(G)+4,3)$-coloring.
\end{corollary}
In this paper we improve the previous result by showing the following results for low maximum average degree.

\begin{theorem}\label{SHAmad}
Let $G$ be a graph with
\begin{enumerate}
\item \label{T1} $\rm{mad}(G)<4$, then $\chi_i(G)\leq \Delta(G)+3$ for every $\Delta(G)\geq 7$ .
\item \label{T2} $\rm{mad}(G)<\frac{9}{2}$, then $\chi_i(G)\leq \Delta(G)+4$ for every $\Delta(G)\geq 9$.
\item \label{T3} $\rm{mad}(G)<5$, then $\chi_i(G)\leq \Delta(G)+5$ for every $\Delta(G)\geq 9$ or $\Delta(G)\leq 5$.
\item \label{T4} $\rm{mad}(G)<5$, then $\chi_i(G)\leq \Delta(G)+6$ for every $6\leq \Delta(G)\leq 8$.
\item \label{T5} $\rm{mad}(G)<6$, then $\chi_i(G)\leq \Delta(G)+6$ for every $\Delta(G)\geq 12$ or $\Delta(G)\leq 6$.
\item \label{T6} $\rm{mad}(G)<6$, then $\chi_i(G)\leq \Delta(G)+7$ for every $7\leq \Delta(G)\leq 11$.
\end{enumerate}	

\end{theorem}

\medskip

For planar graphs, Hosseini Dolama \emph{et al.}~\citep{hdsz04} proved in 2004 the following result: 
\begin{theorem}[\label{DolEr}Hosseini Dolama and Sopena~\cite{hdsz04}]
Every  planar graph $G$ admits an incidence $(\Delta+7)$-coloring. 
\end{theorem}

By using the link between the incidence chromatic number, the star arboricity and the chromatic index of a graph, Yang proved the following theorem:

\begin{theorem}[\label{Yang}Yang~\citep{Ya12}]
For every planar graph $G$,  $\chi_i(G) \leq\Delta(G)+5$, if $\Delta(G)\neq 6$ and  $\chi_i(G)\leq 12$, if $\Delta(G)= 6$.
\end{theorem}

As every planar graph with girth $g$ satisfies $\rm{mad}(G)< \frac{2g}{g-2}$, the following two corollaries can be derived from Theorem~\ref{SHAmad}.\ref{T1} and Theorem~\ref{SHAmad}.\ref{T2}. Along the way, we improve the bound given in Theorem \ref{Yang} for triangle-free planar graphs:

\begin{corollary}\label{corolSHAmad}
Let $G$ be a triangle-free planar graph with $\Delta(G)\geq 7$. Then, $\chi_i(G) \le \Delta(G)+3$. 
\end{corollary}
\begin{corollary}\label{corolSHAmadbis}
Let $G$ be a  planar graph with  with girth $g \geq5$ and $\Delta(G)\geq 9$. Then, $\chi_i(G) \le \Delta(G)+4$. 
\end{corollary}

Moreover if $G$ is a planar graph, then $\rm{mad}(G)<6$. We deduce another proof of Theorem~\ref{DolEr}. 
\begin{corollary}\label{corolSHAmadplanaire}
Let $G$ be a planar graph. Then, $\chi_i(G) \le \Delta(G)+7$. 
\end{corollary}
We have also a better bound for large maximum degree.
\begin{corollary}\label{corolSHAmadpla}
Let $G$ be a planar graph with $\Delta(G)\geq 12$. Then, $\chi_i(G) \le \Delta(G)+6$. 
\end{corollary}
\medskip

However, Theorem~\ref{th:BH1} can be proved with a relatively simple discharging argument, and gives hope that further development of more complicated reducible configurations and discharging rules might unlock generic results with consequences beyond purely extremal.

\medskip

Before proving our results we introduce some notation.

\paragraph{Notation}
Let $G$ be a graph. Let $d(v)$ denote the degree of a vertex $v$ in $G$. A vertex of degree $k$ is called a $k$-vertex. A $k^+$-vertex (respectively, $k^-$-vertex) is a vertex of degree at least $k$ (respectively, at most $k$). A $(l_1,\cdots,l_k)$-vertex is a $k$-vertex having $k$-neighbors $x_1,\cdots,x_k$ such that $d(x_i)=l_i$ for $i\in \{1,\cdots,k\}$.
We will also use for $l_i$ the notation $l_i^+$ (respectively
$l_i^-$), if $x_i$ is  a vertex of degree at least $l_i$ (respectively at most $l_i$).

\section{Proof of Theorem~\ref{th:BH1}}
	
\subsection{Structural property}

Let $d \in \mathbb{N}$, $k \in \mathbb{N}$ and $\alpha \in \mathbb{R}^+$. We consider $\mathcal{G}_d$ the class of graphs with maximum degree at most $d$, and $\mathcal{H}_{d,\alpha} \subseteq \mathcal{G}_d$ the subset of graphs which are not incidence $(d+(1+\alpha) k-1,(1+\alpha) k-1)$-colorable. 

\begin{lemma}\label{lem:BH1}
If $G$ is a minimal graph in $\mathcal{H}_{d,\alpha}$, then every $u \in V(G)$ with $d(u) \le k-1$ has more than $(1+\alpha) k-d(u)$ neighbors of degree at least $d-k+2$.
\end{lemma}

\begin{proof}
Let $G$ be a minimal counter-example to Theorem~\ref{th:BH1}, and $u \in V(G)$ a witness of it. Let $p=d(u)$, and $v_1,\ldots,v_p$ be the $p$ neighbors of $u$, sorted by decreasing degree. By assumption, for every $p \geq i > (1+\alpha) k-p$, we have $d(v_i)\leq  d-k+1$.
By minimality of the counter-example, the graph $G'= G \setminus \{u\}$ admits an incidence $(d+(1+\alpha) k-1,(1+\alpha) k-1)$-coloring. Let us prove that we can extend the coloring to $G$.

For every $1 \leq i \leq p$, we set $e_i=(u,uv_i)$. For every $ 1 \leq i \le (1+\alpha) k-p$, we set $f_i=(v_i,uv_i)$. For every $ p \geq i > (1+\alpha) k-p$, we set $g_i=(v_i,uv_i)$. When referring to \emph{the $f_i$'s}, we actually refer to the set $\{f_i | 1 \leq i \leq (1+\alpha) k-p\}$. Note that $(1+\alpha) k-p \geq 1$. We can similarly define \emph{the $e_i$'s} and \emph{the $g_i$'s}.

Let us now evaluate how many free colors are available, in the worst case, to color each incidence of $u$. Let $1 \leq i \leq p$, and consider $e_i$. We have to ensure that at most $(1+\alpha) k-1$ different colors appear on weak incidences of $v_i$. Let $S_i$ be the set of colors appearing on strong incidences of $v_i$. We have $|S_i|=d(v_i)-1 \leq d-1$. Since there are $d+(1+\alpha) k-1$ different colors, and by assumption, there is a set $W_i$ of $(1+\alpha) k-1$ different colors, with $W_i \cap S_i= \emptyset$, from which the weak incidences of $v_i$ are colored. Then, $e_i$ might be colored with any element of $W_i$, which makes $(1+\alpha) k-1$ colors available for $e_i$. If $i \le (1+\alpha) k-p$, we consider $f_i$. From the initial $d+(1+\alpha)k-1$ colors, we remove $W_i$ ($(1+\alpha)k-1$ colors) for the weak incidences of $v_i$, and at most $d(v_i)-1 \leq d-1$ for the strong incidences of $v_i$. Therefore, $f_i$ has at least one available color, which by construction does not belong to $W_i$. If $i > (1+\alpha) k-p$, we consider $g_i$. From the initial $d+(1+\alpha)k-1$ colors, we remove $W_i$ ($(1+\alpha)k-1$ colors) for the weak incidences of $v_i$, and $d(v_i)-1 \leq d-k$ for the strong incidences of $v_i$. Therefore, $g_i$ has at least $k$ available colors.

We extend the coloring to $G$ by coloring the $f_i$'s, the $e_i$'s and the $g_i$'s, in that order. Let us now argue why this is possible. Since all $v_i$'s are distinct vertices, no two $(v_i,uv_i)$'s can be adjacent. 

Therefore, we first color independently each $f_i$ with a color that does not belong to $W_i$.

Now, for every $i > (1+\alpha) k-p$, $e_i$ has at least $(1+\alpha)k-1 - ((1+\alpha) k-p)=p-1$ colors available. For every $i \leq (1+\alpha) k-p$, the same analysis goes except that $f_i$ was colored with a color that was not available for $e_i$, which allows for one more color: $e_i$ has at least $(1+\alpha)k-1 - ((1+\alpha) k-p-1)=p$ colors available. Consequently, we can color the $e_i$'s by decreasing index.

Now, each $g_i$ has at least $k -p \geq 1$ colors available, and we color the $g_i$'s independently. Since $p=d(u) \le k-1 \le (1+\alpha)k -1$, the coloring obtained for $G$ is an incidence $(d+(1+\alpha) k,(1+\alpha) k)$-coloring.
\end{proof}

Note that Lemma~\ref{lem:BH1} implies that $G$ contains no vertex with degree at most $\frac{(1+\alpha)k}{2}$.

\subsection{Discharging procedure}


Let us try to find sufficient conditions on $d$ that ensure all minimal graphs in $\mathcal{H}_{d,\alpha}$ have average degree at least $k$. Let $G$ be a minimal graph in $\mathcal{H}_{d,\alpha}$.
\medskip

For this purpose, we use a discharging procedure. We assign a weight $\omega(u)=d(u)-k$ to every vertex $u$ of $G$. If we can redistribute the weight in such a way that every vertex has a non-negative weight, then $ad(G) \geq k$. Note that at the beginning, only vertices of degree $k-1$ or less have negative weight. A naive discharging rule to correct this is as follows:

Let $M, c \in \mathbb{R}$.
\begin{itemize}
\item \textbf{R}: Every vertex of degree at least $M$ gives a weight of $c$ to every incident vertex of degree at most $(k-1)$.
\end{itemize}

We try to find good values of $M$ and $c$ so that no vertex has a negative weight after the discharging procedure. For $R$ to be well-defined and for us to be able to apply Lemma~\ref{lem:BH1} in order to ensure that every vertex of degree at most $k-1$ has enough neighbors of degree at least $M$, we need (\ref{Mdk}).

\begin{equation}
k \leq M \leq d-k+2
\label{Mdk}
\end{equation}

For every vertex of degree at least $M$ to have a non-negative weight after application of $R$, it is sufficient to satisfy $M-k - c \times M \geq 0$, thus (\ref{M}), in case all neighbors are of degree less than $k$ and need to receive $c$.

\begin{equation}
M \geq \frac{k}{1-c}
\label{M}
\end{equation}

We assume from now on that (\ref{Mdk}) and (\ref{M}) are satisfied. Every vertex of degree less than $M$ and at least $k$ is not affected by $R$ and has a constant non-negative weight. Therefore, we only need to look at sufficient conditions for vertices of degree at most $k-1$ to have a non-negative weight. Note that Lemma~\ref{lem:BH1} implies that $G$ contains no vertex with degree $\frac{(1+\alpha)k}{2}$ or less. Because $k$ is always an integer while $(1+\alpha)k$ may not always be, we distinguish the case $\alpha=0$ from the rest. Let $u$ be a vertex of degree $p$ at most $k-1$.

\begin{enumerate}
\item Assume $\alpha=0$. By Lemma~\ref{lem:BH1}, $u$ has more than $k-p$ (thus at least $k-p+1$) neighbors of degree at least $d-k+2 \ge M$. Hence $u$ receives at least $(k-p+1)\times c$. For $u$ to have a non-negative final weight, it suffices to have $p-k+c\times(k-p+1) \geq 0$, thus (\ref{u0}).

\begin{equation}
c \geq 1-\frac{1}{k+1-p}
\label{u0}
\end{equation}

The strongest constraint comes from smallest possible $p$, \emph{i.e.} $\frac{k+1}{2}$ (if $k$ is odd, otherwise $p$ has to be at least $\frac{k}{2}+1>\frac{k+1}{2} $). Therefore, we can replace (\ref{u0}) with (\ref{u0b}).

\begin{equation}
c \geq 1-\frac{2}{k+1}
\label{u0b}
\end{equation}

We know that if $d$, $M$ and $c$ satisfy (\ref{Mdk}), (\ref{M}) and (\ref{u0b}), then $\rm{ad}(G) \geq k$ (where $\rm{ad}(G)$ denotes the average degree of $G$). We can take $c= 1-\frac{2}{k+1}$, $M=\frac{k(k+1)}{2}$. It follows that any $d \geq \frac{k(k+3)}{2}-2$ guarantees the conclusion.

\item Case $\alpha > 0$. By Lemma~\ref{lem:BH1}, $u$ has more than $(1+\alpha)k-p$ neighbors of degree at least $d-k+2 \ge M$. Hence $u$ receives at least $((1+\alpha)k-p)\times c$. For $u$ to have a non-negative final weight, it suffices to have $p-k+c\times((1+\alpha)k-p) \geq 0$, or equivalently $c \geq 1-\frac{\alpha k}{(1+\alpha)k-p}$. The constraint is strongest when $p$ is smallest possible, \emph{i.e.} $\frac{(1+\alpha)k}{2}$, hence (\ref{u1}).

\begin{equation}
c \geq 1-\frac{2 \alpha k}{(1+\alpha)k}= 1 - \frac{2 \alpha}{1+\alpha}
\label{u1}
\end{equation}

We know that if $d$, $M$ and $c$ satisfy (\ref{Mdk}), (\ref{M}) and (\ref{u1}), then $\rm{ad}(G) \geq k$. We can take $c= 1-\frac{2 \alpha}{1+\alpha}$, $M=\frac{(\alpha+1)k}{2 \alpha}$. It follows that any $d \geq \frac{k(3\alpha+1)}{2 \alpha}-2$ guarantees the conclusion.
\end{enumerate}
\hfill $\square$

\bigskip

Before proving Theorem~\ref{SHAmad}, we have to introduce some notations used by Hosseini Dolama \emph{et al.}~\cite{hdsz04}.

\medskip

\begin{definition}[Hosseini Dolama \emph{et al.}~\cite{hdsz04}]\label{def}
Let $G$ be a graph.
A partial incidence coloring $\phi'$ of $G$, is an incidence coloring only defined on some subset $I$ of $I(G)$. For every uncolored incidence $(u,uv)\in I(G)\setminus I$,  $F_G^{\phi'}(u,uv)$ is defined by the set of forbidden colors of $(u,uv)$, that is: \begin{center}$F_G^{\phi'}(u,uv)= \phi'(A_u)\cup\phi'(I_u) \cup\phi'(I_v)$,\end{center} where
$I_u$ is the set of incidences of the form $(u,uv)$ and $A_u$ is the set of incidences of the form $(v,vu)$.
\end{definition}

Note that the $(k,\ell)$-incidence coloring can also be seen in this setting as a coloring $\phi$ of $G$ such that for every vertex $v\in V(G)$, $\mid \phi(A_v)\mid\leq \ell$.

\begin{remark}\label{monotone}
{\rm Every $(k,\ell)$-incidence coloring  of a graph is a $(k',\ell)$-incidence coloring for any $k'>k$.}
\end{remark}

\section{Proof of Theorem~\ref{SHAmad}.\ref{T1}}


\subsection{Structural properties}
We proceed by contradiction. Let $H$ be a counterexample to Theorem~\ref{SHAmad}.\ref{T1} that minimizes $|E(H)|+|V(H)|$. By hypothesis there exists $k\geq \max\{\Delta(G),7\}$ such that $H$ does not admit an incidence $(k+3,3)$-coloring. Let $k\geq \max\{\Delta(G),7\}$ be the smallest integer such that $H$ does not admit an incidence $(k+3,3)$-coloring. By using Remark \ref{monotone} we must have $k= \max\{\Delta(G),7\}$. Moreover by minimality it is easy to see that $H$ is connected.\\
$H$ satisfies the following properties:

\begin{lemma}\label{lemma}
$H$ does not contain:
\begin{enumerate}
\item \label{1v} a 1-vertex,
\item \label{2v} a $2$-vertex,
\item \label{3v} a $3$-vertex adjacent to a $3$-vertex,
\item \label{4v} a $((\Delta-2)^-,(\Delta-1)^-,\Delta^-)$-vertex,
\item \label{5v} a $(3,3,(\Delta-1)^-,\Delta^-)$-vertex,
\item \label{6v} a $(3,4^-,4^-,4^-)$-vertex.
\end{enumerate}
\end{lemma}

\begin{proof}

Each of these 6 cases will be dealt with similarly. First, we suppose by contradiction that the described configuration exists in $H$. Then we consider a graph $H'$  obtained from $H$ by deleting an edge or a vertex from $H$. The graph $H'$ has $\rm{mad}(H')\leq \rm{mad}(H)<4$. Due to the minimality of $H$, the graph $H'$  admits an incidence $(k'+3,3)$-coloring for any $k'\geq \max\{\Delta(H'),7\}$. Since $\Delta(H)\geq \Delta(H')$, the set of integers k' contains the set of integers k. Hence for the value $k'=k$, $H'$ admits an incidence $(k+3,3)$-coloring $\phi'$. Finally, for each case, we will prove a contradiction by extending $\phi'$ to an incidence $(k+3,3)$-coloring $\phi$ of $H$.\\


\begin{enumerate}

  \item Suppose $H$ contains a $1$-vertex $u$ and let $v$ be its unique neighbor in $H$. Consider $H'=H-\{u\}$. $H'$ admits an incidence $(k+3,3)$-coloring $\phi'$. We will extend $\phi'$ to an incidence $(k+3,3)$-coloring $\phi$ of $H$ as follows.\\
Since for all $v \in V(H')$, $\mid\phi'(A_v)\mid\leq 3$, we have $$\mid F_H^{\phi'}(v,vu)\mid = \mid\phi'(I_v)\cup\phi'(A_v)\cup\phi'(I_u)\mid\leq\Delta(H)-1+3+0=\Delta(H)+2\leq k+2$$ then there exists at least one color, say $\alpha$, such that $\alpha \notin F_H^{\phi'}(v,vu)$. Hence, we set $\phi(v,vu)=\alpha$ and one can observe that $\mid\phi'(A_u)\mid=1\le 3$. According to $\mid\phi'(A_v)\mid\leq 3$, it suffices to set $\phi(u,uv)= \beta$ for any color $\beta$ in $\phi'(A_v)$ and we are done.\\
We have extended the coloring, a contradiction.
  \item Suppose $H$ contains a $2$-vertex $v$ and let $u$, $w$ be the two neighbors of $v$ in $H$.  By minimality of $H$, $H'= H\setminus \{v\}$ admits an incidence $(k+3,3)$-coloring $\phi'$. We will extend $\phi'$ to an incidence $(k+3,3)$-coloring $\phi$ of $H$ as follows. \\
  $$\mid F_H^{\phi'}(w,wv)\mid = \mid\phi'(I_w)\cup\phi'(A_w)\cup\phi'(I_v)\mid\leq\Delta(H)-1+3+0=\Delta(H)+2\leq k+2$$ Hence there is one available color to color $(w,wv)$, let $\phi(w,wv)=\beta$. By doing the same calculation it is easy to see that there is one available color to color $(u,uv)$, let $\phi(u,uv)=\alpha$. We have $\mid\phi(A_v)\mid=2 \le 3$. Consider now the following cases:
\begin{enumerate}
\item If $\mid\phi'(A_u)\mid=3$ then we color $(v,vu)$ with a color $\gamma \in \phi(A_u)\setminus\{\beta\}$ 
\item If $\mid\phi'(A_u)\mid \le 2$ then we color $(v,vu)$ with a color $\gamma \notin F_H^{\phi'}(v,vu)$ different from $\beta$ (note that we have two choices). One can observe that $\mid\phi(A_u)\mid \le 3$.
\end{enumerate}

If $\mid\phi'(A_w)\mid= 3$ then we color $(v,vw)$ with a color $\zeta \in \phi(A_w)\setminus\{\alpha,\gamma\}$ and we have $\mid\phi(A_w)\mid= 3$. If $\mid\phi'(A_w)\mid\le 2$ then we color $(v,vw)$ with a color $\zeta \notin F_H^{\phi'}(v,vw)$ (different from  $\alpha,\gamma$) and we have $\mid\phi(A_w)\mid \le 3$. We have extended the coloring, a contradiction.

  \item Suppose $H$ contains a $3$-vertex $u$ adjacent to a $3$-vertex $v$. By minimality of $H$, $H'=H\setminus \{uv\}$ has an incidence $(k+3,3)$-coloring $\phi'$. We will extend $\phi'$ to an incidence $(k+3,3)$-coloring $\phi$ of $H$ as follows. 
  
  \begin{center}
 $\mid \phi'(A_u)\mid \leq 2$ and $\mid \phi'(A_v)\mid \leq 2$ 
  \end{center}
  $$\mid F_H^{\phi'}(u,uv)\mid = \mid\phi'(I_u)\cup\phi'(A_u)\cup\phi'(I_v)\mid\leq 2+2+2=6$$ 
We have at least $k-3 \geq 4$ free colors for $(u,uv)$. Choose a color for $(u,uv)$, then for $(v,vu)$ by using the same calculation at most 7 colors are forbidden for this incidence. We have at least $k-4 \geq 3$ free colors. That is more than enough to extend the coloring, a contradiction.

\item Suppose $H$ contains a $((\Delta-2)^-,(\Delta-1)^-,\Delta^-)$-vertex $v$ and let $u_1$, $u_2$ and $u_3$ be the three neighbors of $v$ in $H$ such that $d(u_1)\leq \Delta-2$, $d(u_2)\leq \Delta-1$ and $d(u_3)\leq \Delta$. Consider $H'=H-\{vu_1\}$. By minimality of $H$, $H'$ admits an incidence $(k+3,3)$-coloring $\phi'$. We will extend $\phi'$ to an incidence $(k+3,3)$-coloring $\phi$ of $H$ as follows.\\
\begin{enumerate}
 \item[$(1)$] If $\mid\phi'(A_{u_1})\mid\leq2$, then $\mid F_H^{\phi'}(v,vu_1)\mid = \mid\phi'(I_v)\cup\phi'(A_v)\cup\phi'(I_{u_1})\mid\leq 2+2 +\Delta(H)-3=\Delta(H)+1\leq k+1$ then there exists at least one color, say $\alpha$, such that $\alpha \notin F_H^{\phi'}(v,vu_1)$. Hence, we set $\phi(v,vu_1)=\alpha$ and we have $\mid\phi(A_{u_1})\mid\le 2+1=3$. For the incidence $(u_1,u_1v)$, we have $\mid F_H^{\phi'}(u_1,u_1v)\mid = \mid\phi'(I_{u_1})\cup\phi'(A_{u_1})\cup\phi'(I_v)\mid\leq \Delta(H)-3+3+2=\Delta(H) +2\leq k+2$. Then there exists at least one color $\beta \notin F_H^{\phi'}({u_1},{u_1}v)$. Hence we color the incidence $(u_1,u_1v)$ with $\beta$.
     
 \item [$(2)$] If $\mid\phi'(A_{u_1})\mid=3$. We distinguish two cases :
 \begin{itemize}
  \item[$Case~1.$] Suppose $\phi'(A_{u_1})\nsubseteq\phi'(I_v)\cup\phi'(A_v)$, then we have a color $\alpha \in \phi'(A_{u_1})\setminus\{\phi'(I_v)\cup\phi'(A_v)\}$. We color $(v,vu_1)$ with $\alpha$. 
  $$\mid F_H^{\phi'}(u_1,u_1v)\mid = \mid\phi'(I_v)\cup\phi'(A_{u_1})\cup\phi'(I_{u_1})\mid\leq 2+3+\Delta(H)-3=\Delta(H) +2 \leq k+2$$
   Hence, there exists at least one available color $\beta \notin F_H^{\phi'}(u_1,u_1vu)$, we set $\phi(u_1,u_1v)=\beta$ and we have $\mid\phi'(A_v)\mid\le 3$.
  \item[$Case~2.$] Suppose $\phi'(A_{u_1})\subseteq\phi'(I_v)\cup\phi'(A_v)$. W.l.o.g we assume that the available colors for $(v,vu_1)$ of $\phi'(A_{u_1})$ are the colors $1$, $2$ and $3$. We must notice that we have two available colors to color $(u_2,u_2v)$, if we do not take into account the constraints given by $\phi'(I_v)$. Then we have two cases to consider.
    \begin{enumerate}
    \item [$A.$]  $\phi'(v,vu_2)=1$, $\phi'(v,vu_3)=2$, $\phi'(u_2,u_2v)=3$.
        \begin{enumerate}
        \item [$1.$] If we can change the color of $(u_2,u_2v)$ in $\phi'$ and $\phi'(u_3,u_3v)\neq 3$, then we recolor the incidence $(u_2,u_2v)$ and we have $3\notin \phi'(I_v)\cup\phi'(A_v)$, then we proceed exactly as $case~1$.
        \item [$2.$] If we can change the color of $(u_2,u_2v)$ and $\phi'(u_3,u_3v)= 3$, then we recolor the incidence $(u_2,u_2v)$ with a color $a$ and it easy to see that $a\neq 1$ and $2$. The list of the available colors of $(v,vu_2)$ is $L=\phi'(A_{u_2})=\{1,\alpha,\beta\}$ and $3$ and $a$ are not belonging to $L$. We recolor $(v,vu_2)$ with a color different from $1$, $2$ and we have $1\notin \phi'(I_v)\cup\phi'(A_v)$, then we proceed exactly as $case~1$.
        \item [$3.$] If we cannot change the color of $(u_2,u_2v)$, it means that the available colors for $(u_2,u_2v)$  are $3$ and $2$. The  list of available colors for $(v,vu_2)$ is $L=\phi'(A_{u_2})=\{1,\alpha,\beta\}$ and $3$, $2$ are not belonging to $L$. We color $(v,vu_2)$ with  a color different from $1$ and $\phi'(u_3,u_3v)$. We have $1\notin \phi'(I_v)\cup\phi'(A_v)$, then we proceed exactly as $case~1$.
        \end{enumerate}
    \item [$B.$] $\phi'(v,vu_2)=1$, $\phi'(u_2,u_2v)=2$, $\phi'(u_3,u_3v)=3$ and $\phi'(v,vu_3)=b_3$ such that $b_3\notin \{1,2,3\}$.
       \begin{enumerate}
        \item [$1.$] If we can change the color of $(u_2,u_2v)$, then we recolor the incidence $(u_2,u_2v)$ and we have $2\notin \phi'(I_v)\cup\phi'(A_v)$, then we proceed exactly as $case~1$.
        \item [$2.$] If we cannot change the color of $(u_2,u_2v)$, it means that the available colors for $(u_2,u_2v)$  are $2$ and $b_3$. The  list of available colors for $(v,vu_2)$ is $L=\phi'(A_{u_2})=\{1,\alpha,\beta\}$ and $b_3$, $2$  are not belonging to $L$. We color $(v,vu_2)$ with  a color different from $1$ and $3$, we have $1\notin \phi'(I_v)\cup\phi'(A_v)$, then we proceed exactly as $case~1$.
        \end{enumerate} 
   \end{enumerate} 
 \end{itemize}
 \end{enumerate}
We have extended the coloring, a contradiction.

 
\item Assume that $H$ contains a $(3,3,(\Delta-1)^-,\Delta^-)$-vertex $u$. Let $u_1,u_2$ be the two neighbors of $u$ having a degree equal to 3. Let $v$ be the vertex of degree $\Delta^-$ and $w$ the neighbor of degree $(\Delta-1)^-$. We consider $H'=H \setminus \{u\}$. By minimality of $H$, $H'$ has an incidence $(k+3,3)$-coloring $\phi'$. By using the same computation as above we have: 

\begin{enumerate}
\item at least one free color for $(v,vu)$.
\item at least  $2$ free colors for $(w,wu)$.
\item at least  $k-1\geq 6$ free colors for $(u_i,u_iu)$, $i\in\{1,2\}$. We denote by $L_i$ the list of available colors  of  $(u_i,u_iu)$, $i\in\{1,2\}$. 
\item at least  $k+1 \geq 8$ free colors for $(u,uu_i)$, $i\in\{1,2\}$. 
\item a set of 3 free colors for $(u,uv)$ and a set of 3 free colors for $(u,uw)$ ($\phi'(A_{v})$ and $\phi'(A_{w}))$.

\end{enumerate}
Now we extend the coloring.\\
 We first color $(v,vu)$. We have one free color $\alpha$ for $(v,vu)$. We color $(v,vu)$ with $\alpha$.\\
Since $\mid L_1\mid=\mid L_2\mid=k-1$, there exists 
$\beta \in L_1\cap L_2$. We color $(u_i,u_iu)$, $i\in\{1,2\}$ with $\beta$. It is easy to see that in the new coloring we will have $\mid \phi(A_u) \mid \leq 3$.

Now we extend the coloring in the following order.

\begin{enumerate}
\item We color $(u,uw)$ with a color $\phi(u,uw)$ different from $\alpha$ and $\beta$.
\item We color $(u,uv)$ with a color  $\phi(u,uv)$ different from $\beta$ and $\phi(u,uw)$.
\item We color $(w,wu)$ with a color $\phi(w,wu)$ different from $\phi(u,uv)$.
\item We color $(u,uu_1)$ with a color $\phi(u,uu_1)$ different from 
 $\alpha$, $\phi(u,uw)$,  $\phi(u,uv)$, $\phi(w,wu)$.
 \item We color $(u,uu_2)$ with a color $\phi(u,uu_2)$ different from 
 $\alpha$, $\phi(u,uw)$,  $\phi(u,uv)$, $\phi(w,wu)$ and  
  $\phi(u,uu_1)$.
\end{enumerate}

 Hence we extend the coloring, a contradiction. This completes the proof.

 \item Suppose  $H$ contains a $(3,4^-,4^-,4^-)$-vertex $u$. Let $v$ be the neighbor of degree $3$, $u_1,u_2,u_3$ be the 3 neighbors of degree $4^-$.
 We consider $H'=H \setminus \{u\}$. By minimality of $H$, $H'$ has an incidence $(k+3,3)$-coloring $\phi'$.
  By an easy computation as above we have the following: 
  \begin{enumerate}
\item Let $L(v)$ be the list of free colors for $(v,vu)$.
Then $\mid L(v)\mid=k-1\geq 6$.

\item 
We denote by $L_i$ the list of available colors of each $(u_i,u_iu)$ ($i\in \{1,2,3\}$). Then $\mid L_i\mid=k-3\geq 4$.
\item for each $i\in \{1,2,3\}$ we have a set of 3 available colors to color $(u,uu_i)$ ($\phi'(A_{u_i})$ for $i\in \{1,2,3\}$).
\item We have $k +1\geq 8$ available colors to color $(u,uv)$.
\end{enumerate}
Since $\mid L_i\mid=k-3$ for $i\in \{1,2,3\}$. 
There exists at least two lists having an element in commun. W.l.o.g. assume that $\alpha \in L_1\cap L_2$. We color $(u_1,u_1u)$ and $(u_2,u_2u)$ with $\alpha$. It follows that  in  any cases $\mid\phi(A_v)\mid\leq 3$.\\
 We extend now the coloring in the following order:
 \begin{enumerate}
 \item $\phi(u,uu_3)$ will be a color different from $\alpha$.
  \item $\phi(u,uu_2)$ will be a color different from $\phi(u,uu_3)$.
 \item $\phi(u,uu_1)$ will be a color different from $\phi(u,uu_2)$ and $\phi(u,uu_3)$.
 \item $\phi(u_3,u_3u)$ will be a color different from $\phi(u,uu_2)$ and $\phi(u,uu_1)$.
  \item $\phi(v,vu)$ will be a color different from 
  $\phi(u,uu_1)$, $\phi(u,uu_2)$ and $\phi(u,uu_3)$.
 \item $\phi(u,uv)$ will be a color different from $\alpha$, $\phi(v,vu)$, $\phi(u_3,u_3u)$ and $\phi(u,u_i)$, $i\in \{1,2,3\}$.
\end{enumerate}

Hence we extend the coloring  in both cases, a contradiction. This completes the proof.

 
  \end{enumerate}
 
\end{proof}

\subsection{Discharging procedure}
We define the weight function $\omega: V(H) \rightarrow \mathbb{R}$ with $\omega(x)=d(x)-4$. It follows from the hypothesis on the maximum average degree that the total sum of weights is strictly negative. In the next step, we define discharging rules (R1) to (R3) and we redistribute weights and once the discharging is finished, a new weight function $\omega^\ast$ will be produced. During the discharging process the total sum of weights is kept fixed. Nevertheless, we can show that $\omega^\ast(x) \ge 0$ for all $x\,\in\,V(H)$. This leads to the following contradiction:
$$0\; \le \sum_{x\,\in\,V(H)}\;\omega^\ast(x)\;=\;\sum_{x\,\in\,V(H)}\;\omega(x)\;<\;0$$
and hence, this counterexample cannot exist.

The discharging rules are defined as follows:

\begin{enumerate}

\item[(R1)] Every $k$-vertex, for $k\geq 5$, gives $\frac{k-4}{k}$ to each adjacent $3$-vertex.
\item[(R2)] Every $k$-vertex, for $k\geq 5$, gives $\frac{k-4}{k}$ to each adjacent $4$-vertices having neighbors of degree 3.
\item[(R3)] Every $4$-vertex gives uniformly (in equal parts) its its new weight to its neighbors of degree 3. 


\end{enumerate}

\medskip

\noindent
Let $v\,\in\,V(H)$ be a $k$-vertex. By Lemma~\ref{lemma}.\ref{1v} and Lemma~\ref{lemma}.\ref{2v}, $k \ge 3$. Consider the following cases:

\begin{enumerate}
\item[] {\bf Case $\boldsymbol{k=4.}$} Observe that $\omega(v)=0$. 
By Lemma~\ref{lemma}.\ref{5v}, $v$ has at most 2 neighbors both of degree 3.  If $v$ has no neighbor of degree 3, it gets nothing and it gives nothing and $\omega^\ast(v)=0$. We have now to consider 2 cases: 
\begin{enumerate}
\item If it has only one neighbor of degree 3 it has at least one neighbor having a degree at least 5 by Lemma ~\ref{lemma}.\ref{6v}. Then it gets at least $\frac{1}{5}$ and it gives this weight to its neighbor of degree 3. We have $\omega^\ast(v)=0$.
\item If it has two neighbors of degree 3. Then  it has two neighbors with degree at least 7 by Lemma ~\ref{lemma}.\ref{5v}. Then it gets at least $2 \times \frac{3}{7}$, hence it gives at least $\frac{3}{7}$ to each of its neighbors of degree 3. We have $\omega^\ast(v)=0$.

\end{enumerate}
\item[] {\bf Case $\boldsymbol{k=3.}$} Observe that $\omega(v)=-1$.

\begin{enumerate}
\item $v$ has at most one  neighbor of degree 4, if it is the case, its two other neighbors have by  Lemma ~\ref{lemma}.\ref{4v} a degree at least 7.  Hence $\omega^\ast(v)\geq -1 + \frac{1}{5}+ 2 \times \frac{3}{7}= -1 + \frac{37}{35}=\frac{2}{35} > 0$
\item If $v$ has a neighbor of degree 5 by 
 Lemma ~\ref{lemma}.\ref{4v} the two other neighbors have a degree at least 7. Hence, $\omega^\ast(v)\geq -1 + \frac{1}{5}+ 2 \times \frac{3}{7}= -1 + \frac{37}{35}=\frac{2}{35} > 0$
 \item If $v$ has no neighbor of degree less or equal to 5, we have
 $\omega^\ast(v)\geq -1 + 3 \times \frac{1}{3}=0$
 \end{enumerate}
\item[] {\bf Case $\boldsymbol{k\geq 5.}$}
 The vertex $v$ satisfies $\omega^\ast(v)\ge k-4 -k\times\frac{k-4}{k}\geq 0$.

\end{enumerate}

After performing the discharging procedure the new weights of all vertices are positive and therefore, $H$ cannot exist. This completes the proof of Theorem \ref{SHAmad}.\ref{T2}.\\
\begin{remark}
In what follows we will assume, if we need it, that $\mid \phi'(A_v)\mid =3$. 
If it is not the case, it is easy to see that we can complete $\phi'(A_v)$ with colors not belonging to $\phi'(A_v) \cup \phi'(I_v)$ in order to have $\mid \phi'(A_v)\mid =3$.
\end{remark}

\section{Proof of Theorem \ref{SHAmad}.\ref{T2}}
\subsection{Structural properties}
We proceed by contradiction. We will use the same reasoning than in the proof of Theorem \ref{SHAmad}.\ref{T1}. Let $H$ be a counterexample to Theorem \ref{SHAmad}.\ref{T2} that minimizes $|E(H)|+|V(H)|$. By hypothesis there exists $k\geq \max\{\Delta(G),9\}$ such that $H$ does not admit an incidence $(k+4,4)$-coloring. Let $k\geq \max\{\Delta(G),9\}$ be the smallest integer such that $H$ does not admit an incidence $(k+4,4)$-coloring. By using Remark \ref{monotone} we must have $k= \max\{\Delta(G),9\}$. Moreover by minimality it is easy to see that $H$ is connected.\\
$H$ satisfies the following properties:

\begin{lemma}\label{lemma2}
$H$ does not contain:
\begin{enumerate}
\item \label{1v2} a $1$-vertex,
\item \label{2v2} a $2$-vertex,
\item \label{3v2} a $((\Delta-1)^-,\Delta^-,\Delta^-)$-vertex,
\item \label{4v2} a $4$-vertex adjacent to a $4^-$-vertex,
\item \label{45v2} a $(5^-,5^-,5^-,5^-)$-vertex.
\end{enumerate}
\end{lemma}

\begin{proof}
\begin{enumerate}
\item[1-2.] By using the same method as in the proof of Theorem \ref{SHAmad}.\ref{T1}, it is easy to prove the two first item of Lemma \ref{lemma2}. 

\item[3.] Suppose $H$ contains a $((\Delta-1)^-,\Delta^-,\Delta^-)$-vertex $v$ and let $u_1$, $u_2$ and $u_3$ be the three neighbors of $v$ in $H$ such that $d(u_1)\leq \Delta-1$, $d(u_2)\leq \Delta$ and $d(u_3)\leq \Delta$. Consider $H'=H-\{v\}$. By minimality of $H$, $H'$ admits an incidence $(k+4,4)$-coloring $\phi'$. We will extend $\phi'$ to an incidence $(k+4,4)$-coloring $\phi$ of $H$ as follows.\\

We have: 
$$\mid F_H^{\phi'}(u_1,u_1v)\mid = \mid\phi'(I_{u_1})\cup\phi'(A_{u_1})\cup\phi'(I_v)\mid\leq \Delta(H)-2+4+0=\Delta(H) +2\leq k+2$$
Hence there are 2 free colors $a$ and $b$ to color $(u_1,u_1v)$. For the incidences $(u_i,u_iv)$, $i\in \{2,3\}$ we have at least one color for each incidence. W.l.o.g we set $\phi(u_2,u_2v)=c$ and $\phi((u_3,u_3v)=d$. For each $(v,vu_i)$, $i \in\{1,2,3\}$, we have $4$ available colors: the colors of $\phi'(A_{u_i})$.\\
First we color $(u_1,u_1v)$ with the colors $a$. Then we color $(v,vu_2)$ with a color different from $a$ and $d$, we set $\phi(v,vu_2)=e$. We color $(v,vu_3)$ with  a color different from $a$, $c$ and $e$, we set $\phi(v,vu_2)=f$.
If we can color $(v,vu_1)$, we are done. So we cannot color $(v,vu_1)$, it means that $\phi'(A_{u_1})=\{c,d,e,f\}$ (See Figure 1). 

\begin{figure}[ht]
\centering
\includegraphics[width=3in]{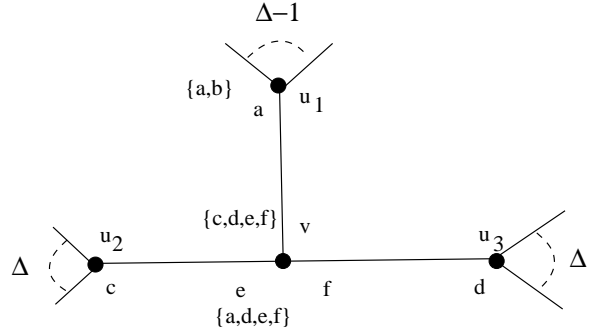}\\
\label{fig1}
\caption{$(\Delta-1,\Delta,\Delta)$-vertex}
\end{figure}

If we can change the color of $(v,vu_2)$ by taking one other available color of $(v,vu_2)$, we change 
the color of $(v,vu_2)$ and we give the color $e$ to $(v,vu_1)$ and we are done. Hence we cannot change the color of $(v,vu_2)$, it means that $\phi'(A_{u_2})=\{a,d,e,f\}$. Then we color $(u_1,u_1v)$ with $b$ ($b$ is not belonging to $\phi'(A_{u_1})=\{c,d,e,f\}$), we color $(v,vu_2)$ with $a$ and $(v,vu_1)$ with $e$, and we are done, a contradiction. This completes the proof.

\item[4.] Suppose $H$ contains $2$ adjacent 4 vertices $u$ and $v$.\\
Consider $H'=H-\{uv\}$. By minimality of $H$, $H'$ admits an incidence $(k+4,4)$-coloring $\phi'$. We will extend $\phi'$ to an incidence $(k+4,4)$-coloring $\phi$ of $H$ as follows.\\
 We have 
$$\mid F_H^{\phi'}(u,uv)\mid = \mid\phi'(I_{u})\cup\phi'(A_{u})\cup\phi'(I_v)\mid\leq 3+3+3=9 \leq k$$
and
$$\mid F_H^{\phi'}(v,vu)\mid = \mid\phi'(I_{v})\cup\phi'(A_{v})\cup\phi'(I_u)\mid\leq 3+3+3=9 \leq k$$
For each incidence we have $4$ free colors.
We can easly extend the coloration, a contradiction. This completes the proof.

\item[5.] Suppose $H$ contains a $(5^-,5^-,5^-,5^-)$-vertex $v$ and let $u_1$, $u_2$, $u_3$ and $u_4$ be the 4 neighbors of $v$ in $H$ such that $d(u_i)=5$, for $i \in\{1,2,3,4\}$. Consider $H'=H-\{v\}$. By minimality of $H$, $H'$ admits an incidence $(k+4,4)$-coloring $\phi'$. \\
 We have 
$$\mid F_H^{\phi'}(u_i,u_iv)\mid = \mid\phi'(I_{u_i})\cup\phi'(A_{u_i})\cup\phi'(I_v)\mid\leq 4+4=8 \leq k-1$$
Hence each incidence $(u_i,u_iv)$ has least $5$ available colors.\\
 Moreover each incidence $(v,vu_i)$ has $4$ available colors ($\phi'(A_{u_i})$).
We will extend $\phi'$ to an incidence $(k+4,4)$-coloring $\phi$ of $H$ as follows.\\
We first color the incidence $(v,vu_i)$ one after the other, next we color the incidences 
$(u_i,u_iv)$. We have $5$ available colors for each  $(u_i,u_iv)$ and $3$ new forbidden colors. This completes the proof.
 
 \end{enumerate}
\end{proof}

\subsection{Discharging procedure}
We define the weight function $\omega: V(H) \rightarrow \mathbb{R}$ with $\omega(x)=d(x)-\frac{9}{2}$. It follows from the hypothesis on the maximum average degree that the total sum of weights is strictly negative. In the next step, we define a discharging rule (R) and we redistribute weights and once the discharging is finished, a new weight function $\omega^\ast$ will be produced. During the discharging process the total sum of weights is kept fixed. Nevertheless, we can show that $\omega^\ast(x) \ge 0$ for all $x\,\in\,V(H)$. This leads to the following contradiction:
$$0\; \le \sum_{x\,\in\,V(H)}\;\omega^\ast(x)\;=\;\sum_{x\,\in\,V(H)}\;\omega(x)\;<\;0$$
and hence, this counterexample cannot exist.\\
We recall:
\begin{enumerate}
\item if $d(v)=3$, $\omega(v)=-\frac{3}{2}$,
\item if $d(v)=4$, $\omega(v)=-\frac{1}{2}$,
\item if $d(v)=5$, $\omega(v)=\frac{1}{2}$.
\end{enumerate}

The discharging rule is defined as follows:

\begin{enumerate}
\item[(R)] Every $k$-vertex ($k\geq 5$) gives $\frac{2k-9}{2k}$ to each of its neighbors having a degree less or equal to 4.
\end{enumerate}
We have to notice that a $k$-vertex do not give more that it has.
\noindent
Let $v\,\in\,V(H)$ be a $k$-vertex. 
By Lemma~\ref{lemma2}.\ref{1v2} and Lemma~\ref{lemma2}.\ref{2v2}, $k \ge 3$. Consider the following cases:

\begin{enumerate}
\item[] {\bf Case $\boldsymbol{k=3.}$} Observe that $\omega(v)=-\frac{3}{2}$. 
By Lemma~\ref{lemma2}.\ref{3v2}, $v$ has neighbors having a degree $\Delta \geq 9$. Hence we have:
$\omega^\ast(v)\geq -\frac{3}{2}+ 3\times\frac{1}{2}=0 $.
\item[] {\bf Case $\boldsymbol{k=4.}$} Observe that by  Lemma~\ref{lemma2}.\ref{4v2} 
and  Lemma~\ref{lemma2}.\ref{45v2}, it has at most 3 neighbors having degree 5. Hence: 
$\omega^\ast(v)\geq -\frac{1}{2}+ 3\times\frac{1}{10} +\frac{3}{12} =\frac{1}{20} \geq 0$.
\item[] {\bf Case $\boldsymbol{k\geq 5.}$} The vertex $v$ satisfies $\omega^\ast(v) \geq k-\frac{9}{2}-k\times\frac{2k-9}{2k}\geq 0$.
\end{enumerate}


After performing the discharging procedure the new weights of all vertices are positive and therefore, $H$ cannot exist. This completes the proof of Theorem \ref{SHAmad}.\ref{T2}.

\section{Proof of Theorem \ref{SHAmad}.\ref{T3}}
\subsection{Structural properties}
We proceed by contradiction. We use the same reasonnig as above. Let $H$ be a counterexample to Theorem \ref{SHAmad}.\ref{T3} that minimizes $|E(H)|+|V(H)|$. By hypothesis there exists $k\geq \max\{\Delta(G),9\}$ such that $H$ does not admit an incidence $(k+5,5)$-coloring. Let $k\geq \max\{\Delta(G),9\}$ be the smallest integer such that $H$ does not admit an incidence $(k+5,5)$-coloring. By using Remark \ref{monotone} we must have $k= \max\{\Delta(G),9\}$. By minimality it is easy to see that $H$ is connected.\\
$H$ satisfies the following properties:
\begin{lemma}\label{lemma3}
$H$ does not contain:
\begin{enumerate}
\item \label{1v3} a $1$-vertex,
\item \label{2v3} a $2$-vertex,
\item \label{3v3} a $3$-vertex,
\item \label{4v3} a $4$-vertex adjacent to a $5^-$-vertex,
\item \label{46v3} a $4$-vertex adjacent to two $6$-vertices.
\end{enumerate}
\end{lemma}
\begin{proof}
\begin{enumerate}
\item[1-2.] By using the same method as in the proof of Theorem \ref{SHAmad}.\ref{T1}, it is easy to prove the two first item of Lemma \ref{lemma3}. 
\item[3.] Suppose $H$ contains a 3-vertex $v$ and let $u_1,u_2,u_3$ be the  neighbors of $v$. Consider $H'=H-\{v\}$. By minimality of $H$, $H'$ admits an incidence $(k+5,5)$-coloring $\phi'$. We will extend $\phi'$ to an incidence $(k+5,5)$-coloring $\phi$ of $H$ as follows.\\
Each incidence $(u_i,u_iv)$ has one avaible color, we color each incidence $(u_i,u_iv)$ with this available color. Each incidence $(v,vu_i)$ has $5$ available colors ($\phi'(A_{u_i})$).We color each incidence $(v,vu_i)$ one after in other to have an incidence coloring. Hence we extend the coloring, a contradition.
\item[4.] Suppose $H$ contains a $4$-vertex $v$ adjacent to a $5^-$-vertex $u$. Consider $H'=H-\{uv\}$. By minimality of $H$, $H'$ admits an incidence $(k+5,5)$-coloring $\phi'$. We will extend $\phi'$ to an incidence $(k+5,5)$-coloring $\phi$ of $H$ as follows.\\
It is easy to see that: 
\begin{enumerate}
\item the incidence $(v,vu)$ has at least $k+5-10=k-5 \geq 4$ free colors,
\item the incidence $(u,uv)$ has at least $k+5-11=k-6 \geq 3$ free colors.
\end{enumerate}
It is easy to see that we can extend the coloring. The proof is left to the reader.
\item[5.] Suppose $H$ contains a 4-vertex $v$ having two neighbors $u_1$ and $u_2$ with 
$d(u_1)=d(u_2)= 6$. Let $u_3$ and $u_4$ be the two other neighbors.
Consider $H'=H-\{v\}$. By minimality of $H$, $H'$ admits an incidence $(k+5,5)$-coloring $\phi'$. We will extend $\phi'$ to an incidence $(k+5,5)$-coloring $\phi$ of $H$ as follows.\\
 By using the same computation as above we have: 

\begin{enumerate}
\item at least one free color for $(u_i,u_iv)$, $i\in \{3,4\}$. We color $(u_i,u_iv)$ with this free color. We set  $\phi(u_3,u_3v)=a$ and  $\phi(u_4,u_4v)=b$.
\item a set of 5 free colors for $(v,vu_i)$, $i\in \{1,2,3,4\}$ ($\phi'(A_{u_i})$, $i\in \{1,2,3,4\}$). We color $(v,vu_1)$ with a color different from $a$ and $b$.
We color $(v,vu_2)$ with a color different from $a$, $b$ and $\phi(v,vu_1)$.
We color $(v,vu_3)$ with a color different from $b$, $\phi(v,vu_1)$ and $(v,vu_2)$.
We color $(v,vu_4)$ with a color different from $a$, $\phi(v,vu_1)$, $\phi(v,vu_2)$ and $\phi(v,vu_3)$.
\item at least  $k-5$ free colors for $(u_i,u_iv)$, $i\in \{1,2\}$, so at least $4$ free colors. We color $(u_1,u_1v)$ with a color different from  $\phi(v,vu_2)$, $\phi(v,vu_3)$ and $\phi(v,vu_4)$ and we color $(u_2,u_2v)$ with a color different from  $\phi(v,vu_1)$, $\phi(v,vu_3)$ and $\phi(v,vu_4)$. 
\end{enumerate}
Hence we have extended the coloration to the whole graph $H$, a contradiction. This completes the proof.
\end{enumerate}
\end{proof}
\subsection{Discharging procedure}
We define the weight function $\omega: V(H) \rightarrow \mathbb{R}$ with $\omega(x)=d(x)-5$. It follows from the hypothesis on the maximum average degree that the total sum of weights is strictly negative. In the next step, we define a discharging rule (R) and we redistribute weights and once the discharging is finished, a new weight function $\omega^\ast$ will be produced. During the discharging process the total sum of weights is kept fixed. Nevertheless, we can show that $\omega^\ast(x) \ge 0$ for all $x\,\in\,V(H)$. This leads to the following contradiction:
$$0\; \le \sum_{x\,\in\,V(H)}\;\omega^\ast(x)\;=\;\sum_{x\,\in\,V(H)}\;\omega(x)\;<\;0$$
and hence, this counterexample cannot exist.\\
We recall that we have only one negative weight: $d(v)=4$, $\omega(v)=-1$.\\

The discharging rule is defined as follows:

\begin{enumerate}
\item[(R)] Every $k$-vertex ($k\geq 6$) gives $\frac{k-5}{k}$ to each of its neighbors having a degree equal to 4.
\end{enumerate}

We have to notice that a $k$-vertex ($k\geq 6$) does not give more that it has.
\noindent
Let $v\,\in\,V(H)$ be a $k$-vertex. By Lemma~\ref{lemma3}.\ref{1v3}, Lemma~\ref{lemma3}.\ref{2v3} and Lemma~\ref{lemma3}.\ref{3v3},   $k \ge 4$. Consider the following cases:

\begin{enumerate}
\item[] {\bf Case $\boldsymbol{k=4.}$} Observe that $\omega(v)=-1$. 
By Lemma~\ref{lemma3}.\ref{4v3} the neighbors of $v$ have a degree at least 6 and by Lemma~\ref{lemma3}.\ref{46v3}, $v$ has at most one  neighbor having a degree equal to $6$. Hence we have:
$\omega^\ast(v)\geq -1+ 3\times\frac{2}{7}+\frac{1}{6} =-1+\frac{43}{42}= \frac{1}{42}\geq 0$.
\item[] {\bf Case $\boldsymbol{k\geq 5.}$} The vertex $v$ satisfies $\omega^\ast(v) \geq k-5-k\times\frac{k-5}{k}\geq 0$.

\end{enumerate}

After performing the discharging procedure the new weights of all vertices are positive and therefore, $H$ cannot exist. This completes the proof of Theorem \ref{SHAmad}.\ref{T3}.

\section{Proof of Theorem \ref{SHAmad}.\ref{T4}}
\subsection{Structural properties}
We proceed by contradiction. We use the same reasoning than prevously. Let $H$ be a counterexample to Theorem \ref{SHAmad}.\ref{T4} that minimizes $|E(H)|+|V(H)|$. By hypothesis there exists $k\geq \max\{\Delta(G),7\}$ such that $H$ does not admit an incidence $(k+6,6)$-coloring. Let $k\geq \max\{\Delta(G),7\}$ be the smallest integer such that $H$ does not admit an incidence $(k+6,6)$-coloring. By using Remark \ref{monotone} we must have $k= \max\{\Delta(G),7\}$. By minimality it is easy to see that $H$ is connected.\\
$H$ satisfies the following properties:
\begin{lemma}\label{lemma4}
$H$ does not contain:
\begin{enumerate}
\item \label{1v4} a $1$-vertex,
\item \label{2v4} a $2$-vertex,
\item \label{3v4} a $3$-vertex,
\item \label{4v4} a $((\Delta-1)^-,\Delta^-,\Delta^-,\Delta^-)$-vertex.
\end{enumerate}
\end{lemma}
\begin{proof}
\begin{enumerate}
\item[1-3.] By using the same method as in the proof of Theorem \ref{SHAmad}.\ref{T1} and Theorem \ref{SHAmad}.\ref{T3}, it is easy to prove the three first item of Lemma \ref{lemma4}. 
\item[4.] Suppose $H$ contains a $((\Delta-1)^-,\Delta^-,\Delta^-,\Delta^-)$-vertex $v$ and let $u_1$, $u_2$, $u_3$ and $u_4$ be the four neighbors of $v$ in $H$ such that $d(u_1)\leq \Delta-1$, $d(u_i)\leq \Delta$ for $i\in \{2,3,4\}$. Consider $H'=H-\{v\}$. By minimality of $H$, $H'$ admits an incidence $(k+6,6)$-coloring $\phi'$. We will extend $\phi'$ to an incidence $(k+6,6)$-coloring $\phi$ of $H$ as follows.\\
By using the same computation as above we have: 

\begin{enumerate}
\item for each $(u_i,u_iv)$, $i\in\{2,3,4\}$, there is at least one free color.
\item at least two free colors for $(u_1,u_1v)$,

\item at least  $6$ free colors for $(v,vu_i)$, $i\in\{1,2,3,4\}$  ($\phi'(A_{u_i})$,$i\in\{1,2,3,4\}$).
\end{enumerate}
We extend the incidence coloring as follows:
\begin{enumerate}
\item We color each  $(u_i,u_iv)$, $i\in\{2,3,4\}$ with its free color. We set: 
$\phi(u_2,u_2v)=a$, $\phi(u_3,u_3v)=b$ and $\phi(u_4,u_4v)=c$.
\item W.l.o.g. we assume that $\{d,e\}$ are the available colors for $(u_1,u_1v)$. We set $\phi(u_1,u_1v)=d$.
\item We color $(v,vu_2)$ with a color different from $b$, $c$ and $d$, we set $\phi(v,vu_2)=1$.
\item We color $(v,vu_3)$ with a color different from $a$, $c$, $d$ and $\phi(v,vu_2)=1$. We set $\phi(v,vu_3)=2$.
\item We color $(v,vu_4)$ with a color different from $a$, $b$, $d$, $\phi(v,vu_2)=1$ and $\phi(v,vu_3)=2$. We set $\phi(v,vu_4)=3$.
\item If we have one available color for $(v,vu_1)$ among its 6 free colors. We are done.
\item If we cannot color $(v,vu_1)$, it means that the list of free colors of  $(v,vu_1)$ is $\{1,2,3,a,b,c\}=\phi'(A_{u_1})$. If we can take an other free color for $(v,vu_2)$, without destroying our incidence coloring, we change the color of $(v,vu_2)$ and color $(v,vu_1)$ with $1$. Hence we cannot change the color of $(v,vu_2)$. So the set of free colors of  $(v,vu_2)$ is $\{1,2,3,b,c,d\}=\phi'(A_{u_2})$.
\item We set  $\phi(u_1,u_1v)=e\notin \{1,2,3,a,b,c\}=\phi'(A_{u_1})$, $\phi(v,vu_2)=d$ and $\phi(v,vu_1)=1$ and we are done. 
\end{enumerate}
We have extended the coloring to $H$, a contradiction. This completes the proof.


\end{enumerate}
\end{proof}
\subsection{Discharging procedure}
We define the weight function $\omega: V(H) \rightarrow \mathbb{R}$ with $\omega(x)=d(x)-5$. It follows from the hypothesis on the maximum average degree that the total sum of weights is strictly negative. In the next step, we define a discharging rule (R) and we redistribute weights and once the discharging is finished, a new weight function $\omega^\ast$ will be produced. During the discharging process the total sum of weights is kept fixed. Nevertheless, we can show that $\omega^\ast(x) \ge 0$ for all $x\,\in\,V(H)$. This leads to the following contradiction:
$$0\; \le \sum_{x\,\in\,V(H)}\;\omega^\ast(x)\;=\;\sum_{x\,\in\,V(H)}\;\omega(x)\;<\;0$$
and hence, this counterexample cannot exist.\\
We recall that we have only one negative weight: $d(v)=4$, $\omega(v)=-1$.\\

The discharging rule is defined as follows:

\begin{enumerate}
\item[(R)] Every $k$-vertex ($k\geq 6$) gives $\frac{k-5}{k}$ to each of its neighbors having a degree equal to 4.
\end{enumerate}

We have to notice that a $k$-vertex ($k\geq 6$) does not give more that it has.
\noindent
Let $v\,\in\,V(H)$ be a $k$-vertex. By Lemma~\ref{lemma4}.\ref{1v4}, Lemma~\ref{lemma4}.\ref{2v4} and Lemma~\ref{lemma4}.\ref{3v4}, $k \ge 4$. Consider the following cases:

\begin{enumerate}
\item[] {\bf Case $\boldsymbol{k=4.}$} Observe that $\omega(v)=-1$. 
By Lemma~\ref{lemma4}.\ref{4v4}, $v$ has four  neighbors of degree equal to $\Delta \geq 7$. Hence we have:
$\omega^\ast(v)\geq -1+ 4\times\frac{2}{7}= \frac{1}{7}\geq 0$.

\item[] {\bf Case $\boldsymbol{k\geq 5.}$} The vertex $v$ satisfies $\omega^\ast(v) \geq k-5-k\times\frac{k-5}{k}\geq 0$.
\end{enumerate}

After performing the discharging procedure the new weights of all vertices are positive and therefore, $H$ cannot exist. Hence we have proved the first sentance of Theorem \ref{SHAmad}.\ref{T4}. Now by Theorem \ref{borne} if $\Delta \leq 6$, $\chi_i(G)\leq 2\Delta \leq \Delta +6$. This completes the proof of Theorem \ref{SHAmad}.\ref{T4}.

\section{Proof of Theorem \ref{SHAmad}.\ref{T5}}
\subsection{Structural properties}
We proceed by contradiction. We use the same reasoning than prevously. Let $H$ be a counterexample to Theorem \ref{SHAmad}.\ref{T5} that minimizes $|E(H)|+|V(H)|$. By hypothesis there exists $k\geq \max\{\Delta(G),12\}$ such that $H$ does not admit an incidence $(k+6,6)$-coloring. Let $k\geq \max\{\Delta(G),12\}$ be the smallest integer such that $H$ does not admit an incidence $(k+6,6)$-coloring. By using Remark \ref{monotone} we must have $k= \max\{\Delta(G),12\}$. By minimality it is easy to see that $H$ is connected.\\
$H$ satisfies the following properties:
\begin{lemma}\label{lemma5}
$H$ does not contain:
\begin{enumerate}
\item \label{1v5} a $1$-vertex,
\item \label{2v5} a $2$-vertex,
\item \label{3v5} a $3$-vertex,
\item \label{4v5} a $((\Delta-1)^-,\Delta^-,\Delta^-,\Delta^-)$-vertex,
\item \label{58v5} a $(8^-,8^-,8^-,\Delta^-,\Delta^-)$-vertex.
\end{enumerate}
\end{lemma}
\begin{proof}
\begin{enumerate}
\item[1-3.] By using the same method as in the proof of Theorem \ref{SHAmad}.\ref{T1} and Theorem \ref{SHAmad}.\ref{T3}, it is easy to prove the three first item of Lemma \ref{lemma5}. 

\item[4.] Suppose $H$ contains a $((\Delta-1)^-,\Delta^-,\Delta^-,\Delta^-)$-vertex $v$. We proceed as in the proof of Theorem \ref{SHAmad}.\ref{T4}, Lemma \ref{lemma4}.\ref{4v4}.


\item[5.] Suppose $H$ contains a $(8^-,8^-,8^-,\Delta^-,\Delta^-)$-vertex $v$. Let $u_1$, $u_2$, $u_3$ be the $3$ neighbors of $v$ in $H$ having a degree equal to 8, let $u_4$ and $u_5$ be the two neighbors of $v$ in $H$ such that $d(u_4)\leq \Delta$, $d(u_5)\leq \Delta$. Consider $H'=H-\{v\}$. By minimality of $H$, $H'$ admits an incidence $(k+6,6)$-coloring $\phi'$. We will extend $\phi'$ to an incidence $(k+6,6)$-coloring $\phi$ of $H$ as follows.\\
By using the same computation as above we have: 
\begin{enumerate}
\item For each incidence $(u_i,u_iv)$,  $i\in\{1,2,3\}$ 
$$\mid F_H^{\phi'}(u_i,u_iv)\mid = \mid\phi'(I_{u_i})\cup\phi'(A_{u_i})\cup\phi'(I_v)\mid\leq 7+6+0=13$$
It implies that we have at least $5$ free colors for each incidence $(u_i,u_iv)$,  $i\in\{1,2,3\}$.
\item for each $(u_i,u_iv)$, $i\in\{4,5\}$, there is at least one free color.
\item at least  $6$ free colors for $(v,vu_i)$, $i\in\{1,2,3,4,5\}$  ($\phi'(A_{u_i})$, $i\in\{1,2,3,4,5\}$).
\end{enumerate}
We extend the incidence coloring as follows:
\begin{enumerate}
\item First we color the incidence $(u_i,u_iv)$, $i\in\{4,5\}$ with the free color, we set $\phi(u_4,u_4v)=a$ and $\phi(u_5,u_5v)=b$.
\item We color $(v,vu_1)$ with a color $\phi(v,vu_1)$ different from $a$ and $b$.
\item We color $(v,vu_2)$ with a color $\phi(v,vu_2)$ different from $a$, $b$ and $\phi(v,vu_1)$.
\item We color $(v,vu_3)$ with a color $\phi(v,vu_3)$ different from $a$, $b$, $\phi(v,vu_1)$ and $\phi(v,vu_2)$.
\item We color $(v,vu_4)$ with a color $\phi(v,vu_4)$ different from  $b$, $\phi(v,vu_1)$, $\phi(v,vu_2)$ and  $\phi(v,vu_3)$.
\item We color $(v,vu_5)$ with a color $\phi(v,vu_5)$ different from  $a$, $\phi(v,vu_1)$, $\phi(v,vu_2)$, $\phi(v,vu_3)$ and  $\phi(v,vu_4)$.
\item We color each $(u_i,u_iv)$,  $i\in\{1,2,3\}$ with a color different from the $4$ new forbidden colors incident to $v$.
\end{enumerate}
We have extended the coloring to $H$, a contradiction. This completes the proof.
\end{enumerate}
\end{proof}
\subsection{Discharging procedure}
We define the weight function $\omega: V(H) \rightarrow \mathbb{R}$ with $\omega(x)=d(x)-6$. It follows from the hypothesis on the maximum average degree that the total sum of weights is strictly negative. In the next step, we define a discharging rule (R) and we redistribute weights and once the discharging is finished, a new weight function $\omega^\ast$ will be produced. During the discharging process the total sum of weights is kept fixed. Nevertheless, we can show that $\omega^\ast(x) \ge 0$ for all $x\,\in\,V(H)$. This leads to the following contradiction:
$$0\; \le \sum_{x\,\in\,V(H)}\;\omega^\ast(x)\;=\;\sum_{x\,\in\,V(H)}\;\omega(x)\;<\;0$$
and hence, this counterexample cannot exist.\\
We recall that we have : 
\begin{itemize}
\item $d(v)=4$, $\omega(v)=-2$
\item $d(v)=5$, $\omega(v)=-1$
\end{itemize}
The discharging rule is defined as follows:

\begin{enumerate}
\item[(R)] Every $k$-vertex ($k\geq 7$) gives $\frac{k-6}{k}$ to each of its neighbors having a degree equal to 4 or 5.
\end{enumerate}

We have to notice that a $k$-vertex ($k\geq 7$) does not give more that it has.
\noindent
Let $v\,\in\,V(H)$ be a $k$-vertex. By Lemma~\ref{lemma5}.\ref{1v5}, Lemma~\ref{lemma5}.\ref{2v5} and Lemma~\ref{lemma5}.\ref{3v5}, $k \ge 4$. Consider the following cases:

\begin{enumerate}
\item[] {\bf Case $\boldsymbol{k=4.}$} Observe that $\omega(v)=-2$. 
By Lemma~\ref{lemma5}.\ref{4v5}, $v$ has four  neighbors of degree equal to $\Delta \geq 12$. Hence we have:
$\omega^\ast(v)\geq -2+ 4\times\frac{1}{2}= 0$.
\item[] {\bf Case $\boldsymbol{k= 5.}$} $\omega^\ast(v)\geq 0$.
Observe that $\omega(v)=-1$. 
By Lemma~\ref{lemma5}.\ref{58v5}, $v$ has at least 3  neighbors of degree equal to $9$. Hence we have:
$\omega^\ast(v)\geq -1+ 3\times\frac{1}{3}= 0$.
\item[] {\bf Case $\boldsymbol{k\geq 6.}$} The vertex $v$ satisfies $\omega^\ast(v) \geq k-6-k\times\frac{k-6}{k}\geq 0$.
\end{enumerate}

After performing the discharging procedure the new weights of all vertices are positive and therefore, $H$ cannot exist. This completes the proof of Theorem \ref{SHAmad}.\ref{T5}.

\section{Proof of Theorem \ref{SHAmad}.\ref{T6}}
\subsection{Structural properties}
We proceed by contradiction. Let $H$ be a counterexample to Theorem \ref{SHAmad}.\ref{T6} that minimizes $|E(H)|+|V(H)|$. By hypothesis there exists $k\geq \max\{\Delta(G),8\}$ such that $H$ does not admit an incidence $(k+7,7)$-coloring. Let $k\geq \max\{\Delta(G),8\}$ be the smallest integer such that $H$ does not admit an incidence $(k+7,7)$-coloring. By using Remark \ref{monotone} we must have $k= \max\{\Delta(G),8\}$. By minimality it is easy to see that $H$ is connected.\\
$H$ satisfies the following properties:
\begin{lemma}\label{lemma6}
$H$ does not contain:
\begin{enumerate}
\item \label{1v6} a $1$-vertex,
\item \label{2v6} a $2$-vertex,
\item \label{3v6} a $3$-vertex,
\item \label{4v6} a $4$-vertex,
\item \label{5v6} a $5$-vertex adjacent to a $6^-$-vertex,
\item \label{58v6} a $(5,5,5,5,5,5,\Delta^-)$-vertex.
\end{enumerate}
\end{lemma}
\begin{proof}
\begin{enumerate}
\item[1-3.] By using the same method as in the proof of Theorem \ref{SHAmad}.\ref{T1} and Theorem \ref{SHAmad}.\ref{T3}, it is easy to prove the three first item of Lemma \ref{lemma6}. 

\item[4.] Suppose $H$ contains a 4-vertex $v$ and let $u_1,u_2,u_3,u_4$ be the neighbors of $v$. Consider $H'=H-\{v\}$. By minimality of $H$, $H'$ admits an incidence $(k+7,7)$-coloring $\phi'$. We will extend $\phi'$ to an incidence $(k+7,7)$-coloring $\phi$ of $H$ as follows.\\
Each incidence $(u_i,u_iv)$ $i \in \{1,2,3,4\}$ has one avaible color, we color each incidence $(u_i,u_iv)$ with this available color. Each incidence $(v,vu_i)$ has $7$ available colors ($\phi'(A_{u_i})$). We color each incidence $(v,vu_i)$ one after the other in order to have an incidence coloring. Hence we extend the coloring, a contradition.
\item[5.] Suppose $H$ contains a $5$-vertex $v$ adjacent to a $6^-$-vertex $u$. Consider $H'=H-\{uv\}$. By minimality of $H$, $H'$ admits an incidence $(k+7,7)$-coloring $\phi'$. We will extend $\phi'$ to an incidence $(k+7,7)$-coloring $\phi$ of $H$ as follows.\\
We have:
$$\mid F_H^{\phi'}(u,uv)\mid = \mid\phi'(I_{u})\cup\phi'(A_{u})\cup\phi'(I_v)\mid\leq 5+5+4=14$$
Since we have at least 15 colors, we have at least one available color for $(u,uv)$.
Moreover we allow by hypothesis 7 colors for $\phi(A_{v})$ and $\mid\phi'(A_{v})\mid \leq 4$. Hence we can color $(u,uv)$ with the free color that it has. Since $\mid\phi'(A_{v})\mid\leq 5$, there are at least 1 free colors for the incidence $(v,vu)$.\\
It is then easy to extend the coloring to the whole graph $H$. A contradiction.

\item[6.] Suppose $H$ contains a $(5,5,5,5,5,5,\Delta^-)$-vertex $v$. Let $u_i$, $i \in \{1,2,3,4,5,6\}$ be the $6$ neighbors of $v$ in $H$ having a degree equal to 5, let $u_7$  be the neighbors of $v$ in $H$ such that $d(u_7)\leq \Delta$. Consider $H'=H-\{v\}$. By minimality of $H$, $H'$ admits a $(k+7,7)$-incidence coloring $\phi'$. We will extend $\phi'$ to an incidence $(k+7,7)$-coloring $\phi$ of $H$ as follows.\\
By using the same computation as above we have: 
\begin{enumerate}
\item For each incidence $(u_i,u_iv)$,  $i\in\{1,2,3,4,5,6\}$ 
$$\mid F_H^{\phi'}(u_i,u_iv)\mid = \mid\phi'(I_{u_i})\cup\phi'(A_{u_i})\cup\phi'(I_v)\mid\leq 4+4+0=8$$
It implies that we have at least $7$ free colors for each incidence $(u_i,u_iv)$,  $i\in\{1,2,3,4,5,6\}$.
\item for  $(u_7,u_7v)$, there is at least one free color.
\item at least  $7$ free colors for $(v,vu_i)$, $i\in\{1,2,3,4,5,6,7\}$  ($\phi'(A_{u_i})$, $i\in\{1,2,3,4,5,6,7\}$).
\end{enumerate}
We extend the incidence coloring as follows:
\begin{enumerate}
\item First we color the incidence $(u_7,u_7v)$,  with the free color, we set $\phi(u_7,u_7v)=a$.
\item We color $(v,vu_1)$ with a color $\phi(v,vu_1)$ different from $a$.
\item We color $(v,vu_2)$ with a color $\phi(v,vu_2)$ different from $a$, $\phi(v,vu_1)$.
\item We color $(v,vu_3)$ with a color $\phi(v,vu_3)$ different from $a$, $\phi(v,vu_1)$ and $\phi(v,vu_2)$.
\item We color $(v,vu_4)$ with a color $\phi(v,vu_4)$ different from  $a$, $\phi(v,vu_1)$, $\phi(v,vu_2)$ and  $\phi(v,vu_3)$.
\item We color $(v,vu_5)$ with a color $\phi(v,vu_5)$ different from  $a$, $\phi(v,vu_1)$, $\phi(v,vu_2)$, $\phi(v,vu_3)$ and  $\phi(v,vu_4)$.

\item We color $(v,vu_6)$ with a color $\phi(v,vu_6)$ different from  $a$, $\phi(v,vu_1)$, $\phi(v,vu_2)$, $\phi(v,vu_3)$ and  $\phi(v,vu_4)$ and $\phi(v,vu_5)$.
\item We color $(v,vu_7)$ with a color $\phi(v,vu_7)$ different from $\phi(v,vu_1)$, $\phi(v,vu_2)$, $\phi(v,vu_3)$ and  $\phi(v,vu_4)$ and $\phi(v,vu_5)$ and $\phi(v,vu_6)$.

\item We color each $(u_i,u_iv)$,  $i\in\{1,2,3,5,6\}$ with a color different from the $6$ new forbidden colors incident to $v$.
\end{enumerate}
We have extended the coloring to $H$, a contradiction. This completes the proof.

\end{enumerate}
\end{proof}
\subsection{Discharging procedure}
We define the weight function $\omega: V(H) \rightarrow \mathbb{R}$ with $\omega(x)=d(x)-6$. It follows from the hypothesis on the maximum average degree that the total sum of weights is strictly negative. In the next step, we define a discharging rules (R1), (R2) and we redistribute weights and once the discharging is finished, a new weight function $\omega^\ast$ will be produced. During the discharging process the total sum of weights is kept fixed. Nevertheless, we can show that $\omega^\ast(x) \ge 0$ for all $x\,\in\,V(H)$. This leads to the following contradiction:
$$0\; \le \sum_{x\,\in\,V(H)}\;\omega^\ast(x)\;=\;\sum_{x\,\in\,V(H)}\;\omega(x)\;<\;0$$
and hence, this counterexample cannot exist.\\
We recall that we have only one negative weight: $d(v)=5$, $\omega(v)=-1$.\\

The discharging rules are defined as follows:

\begin{enumerate}
\item[(R1)] Every $7$-vertex  gives $\frac{1}{5}$ to each of its neighbors having a degree equal to  $5$.
\item[(R2)] Every $k$-vertex ($k\geq 8$) gives $\frac{k-6}{k}$ to each of its neighbors having a degree equal to  $5$.
\end{enumerate}

We have to notice that a $k$-vertex ($k\geq 8$) does not give more that it has. Moreover a $7$-vertex has at most $5$ neighbors of degree 5.
\noindent
Let $v\,\in\,V(H)$ be a $k$-vertex. By Lemma~\ref{lemma6}.\ref{1v6}, Lemma~\ref{lemma6}.\ref{2v6}, Lemma~\ref{lemma6}.\ref{3v6} and  Lemma~\ref{lemma6}.\ref{4v6} we have   $k \ge 5$. Consider the following cases:

\begin{enumerate}
\item[] {\bf Case $\boldsymbol{k=5.}$} Observe that $\omega(v)=-1$. 
By Lemma~\ref{lemma6}.\ref{5v6}, $v$ has 5  neighbors of degree greater or equal to 7. Hence we have:
$\omega^\ast(v)\geq -1+ 5\times\frac{1}{5}= 0$.
\item[] {\bf Case $\boldsymbol{k= 6.}$} $\omega(v)=0$, $v$ has weight 0 and gives nothing. $\omega^\ast(v)=0$,

\item[] {\bf Case $\boldsymbol{k= 7.}$} $\omega(v)=1$ and
$\omega^\ast(v)\geq 1 -5\times\frac{1}{5}=0$, $v$ has at most 5 neighbors of degree 5  by Lemma~\ref{lemma6}.\ref{58v6}.
\item[] {\bf Case $\boldsymbol{k\geq 8.}$} $v$ does not give more that it has.
$\omega^\ast(v)\geq 0$.
\end{enumerate}

After performing the discharging procedure the new weights of all vertices are positive and therefore, $H$ cannot exist. Hence we have proved the first sentance of Theorem \ref{SHAmad}.\ref{T6}. Now by Theorem \ref{borne} if $\Delta \leq 7$, we have $\chi_i(G)\leq 2\Delta \leq \Delta +7$. This completes the proof of Theorem \ref{SHAmad}.\ref{T6}.

\begin{remark}
{\rm By using Theorem \ref{borne} or the result of \cite{cubic1}, the result of Theorem \ref{SHAmad}.\ref{T1} is true for $\Delta \leq 3$. More precisely every graph with $\Delta(G)\leq 3$, admits an incidence $(\Delta(G)+3,3)$-coloring.
The question remains open for graphs with maximal degree $4$, $5$ or $6$ and a maximal average degree less than $4$. A similar question can be asked for the other results of Theorem \ref{SHAmad}}.
\end{remark}


\end{document}